\documentclass[letterpaper, 10 pt,conference]{ieeeconf}  

\IEEEoverridecommandlockouts                              

\title{
Bifurcation analysis of an opinion dynamics  model coupled with an environmental dynamics
}

\author{Anthony Couthures$^1$, Anastasia Bizyaeva$^2$, Vineeth S. Varma$^{1,3}$, Alessio Franci$^4$, Irinel-Constantin Mor\u{a}rescu$^{1,3}$
\thanks{The work of V.S. Varma and I.C. Mor\u{a}rescu was supported by project DECIDE, no. 57/14.11.2022 funded under the PNRR I8 scheme by the Romanian Ministry of Research, Innovation, and Digitisation}
\thanks{$^1$Universit\'e de Lorraine, CNRS, CRAN, F-54000 Nancy, France. {\tt\small anthony.couthures@univ-lorraine.fr}}%
\thanks{$^3$associated with Automation Department, Technical University of Cluj-Napoca, Memorandumului 28, 400114 Cluj-Napoca, Romania.}%
\thanks{$^2$Sibley School of Mechanical and Aerospace Engineering at Cornell University}
\thanks{$^4$Department of Electrical Engineering and Computer Science of the University of Liege and WEL Research Institute, Wavre, Belgium.}
}

\usepackage{xcolor}

\usepackage{amssymb}
\usepackage{mathtools}
\usepackage[noadjust]{cite}

\usepackage{amsfonts}
\usepackage{dsfont}
\usepackage{amsmath}
\usepackage{amssymb}
\usepackage{bm}
\usepackage{subcaption}
\usepackage{comment}

\newtheorem{theorem}{Theorem}

\newtheorem{proposition}{Proposition}

\newtheorem{lemma}{Lemma}

\newtheorem{assumption}{Assumption}

\newtheorem{remark}{Remark}

\usepackage{pgfplots}

\usepackage{hyperref}

\newcommand{\mc}{\mathcal}

\newcommand{\Vcal}{\mc{V}}
\newcommand{\Gcal}{\mc{G}}
\newcommand{\Ecal}{\mc{E}}

\newcommand{\discretset}[2]{ \left\{#1, \dots, #2 \right\}}

\newcommand{\Vset}{\discretset{1}{N}}

\newcommand{\norm}[1]{\left\| #1 \right\|}
\newcommand{\abs}[1]{\left| #1 \right|}

\newcommand{\R}{\mathbb{R}}

\newcommand{\one}{\mathbf{1}}

\begin{document}
\bstctlcite{IEEEexample:BSTcontrol}

\maketitle
\thispagestyle{empty}
\pagestyle{empty}

\begin{abstract}
We consider an opinion dynamics model coupled with an environmental dynamics. Based on a forward invariance argument, we can simplify the analysis of the asymptotic behavior to the case when all the opinions in the social network are synchronized. Our goal is to emphasize the role of the trust given to the environmental signal in the asymptotic behavior of the opinion dynamics and implicitly of the coupled system. To do that, we conduct a bifurcation analysis of the system around the origin when the trust parameter is varying. Specific conditions are presented for both pitchfork and Hopf bifurcation. Numerical illustration completes the theoretical findings. 
\end{abstract}

\section{INTRODUCTION}

Modeling the dynamics of climate change and environmental processes is of pressing importance and has received significant attention over recent decades. It is noteworthy that opinion dynamics play an important role in the environmental processes and vice-versa. On one hand, individuals' opinions and behaviors are shaped by social interactions. Such interactions are often modeled via networked or multi-agent systems where consensus, polarization, and other collective phenomena emerge \cite{hegselmann2002opinion,Friedkin,MG10,bizyaevaNonlinearOpinionDynamics2023}. On the other hand, human actions can have a profound and sometimes irreversible impact on the environment. This duality is particularly evident in contexts such as climate change debates, sustainable behavior adoption, and collective decision-making in environmental policy \cite{MMN19,AltafiniLeadership23}. Although Opinion Dynamics (OD)  and environmental processes have been extensively addressed separately, the interplay between the two remains insufficiently explored. Some research directions considering this interaction include the evolutionary game perspective \cite{weitzOscillatingTragedyCommons2016,tilmanEvolutionaryGamesEnvironmental2020} and the dynamical systems one \cite{frieswijkModelingCoevolutionClimate2023,couthuresAnalysisOpinionDynamics2024}.

In this paper, we propose and analyze a coupled model that integrates opinion dynamics with environmental feedback as a continuous-time extension of recent work \cite{couthuresAnalysisOpinionDynamics2024}. Each agent in the network holds a continuously evolving opinion representing, for instance, a spectrum of attitudes from pro-environmental to anti-environmental behavior. Agents update their opinions based on two distinct influences: (i) social interactions with neighbors, mediated by a signal function, and (ii) the perceived state of the environment, which is itself affected by the collective behavior of the agents. The environmental state evolves according to a linear dynamic equation controlled by aggregate opinions, and agents indirectly perceive the environmental condition via a response function. This formulation captures the inherent feedback loop between individual behavior and the state of the environment.

The main contributions of this paper are summarized as follows. We consider a simplified model coupling the opinion and environment dynamics. Mathematically, the model is formulated as a system of ordinary differential equations in which both the agents’ opinions and the environmental state evolve continuously over time. Our first technical result establishes the forward invariance of the synchronization manifold (all the opinions in the social network coincide). Next, under appropriate assumptions, we conduct a detailed analysis of the Fully Synchronized Opinion coupled with the Environment (FSOE) dynamics of the system, characterizing the conditions under which the system exhibits singular points. In the FSOE setup, we conduct a bifurcation analysis demonstrating both pitchfork and Hopf bifurcations emanating from the trivial equilibrium. In other words, we characterize the range of parameters guaranteeing 
 the presence of oscillating behaviors between opinions and environment.

The rest of the paper is organized as follows. In Section~\ref{sec:problem_formulation}, we describe the coupled opinion–environment model along with the setup under consideration. Section~\ref{sec:preliminary_results} states a preliminary result on the forward invariance allowing to reduce the analysis to the FSOE case. Section~\ref{sec:analysis} is devoted to the analysis of the FSOE dynamics, including the characterization of equilibria, the identification of singular points, and the bifurcation analysis. Finally, Section~\ref{sec:numerical_simulations} presents short numerical simulations that illustrate the system's behavior under different parameter settings. We conclude the paper in Section~\ref{sec:conclusions} with a summary of the main results.\\
{\bf Notation} We will denote by $\R$ and $\R_{\geq 0}$ the set of real and non-negative real numbers, respectively. For a vector $\boldsymbol{x} \in \R^N$, we denote by $x_i$ the $i$-th component of $\boldsymbol{x}$. The components of the matrix $\boldsymbol{A} \in \R^{N\times N}$ are denoted $a_{ij}$. The $i$-th vector of the canonical basis of $\R^N$ is $\boldsymbol{e}_i$ and $\one$ is the vector of $\R^N$ with all components equal to $1$. We also use the standard notation $\mathrm{diag}(\boldsymbol{x}) \in \R^{N\times N}$ for the diagonal matrix with diagonal elements given by the vector $\boldsymbol{x} \in \R^{N}$. For a function $f: \mathcal{X} \to \mathcal{X}$, we denote $\mathrm{Fix}(f) = \left\{ x \in \mathcal{X} \mid f(x) = x \right\}$ the set of fixed points of $f$ in $\mathcal{X}$.

\section{Problem formulation}\label{sec:problem_formulation}

We consider the classical multi-agent framework in which $N$ individuals/agents belonging to the set $\Vcal = \Vset$ interact according to an \emph{undirected fixed graph} $\Gcal = \left(\Vcal, \Ecal\right)$. We denote by $\boldsymbol{A} \in \R^{N\times N}$ its \emph{adjacency matrix}, i.e., $a_{ij} = 1$ if $(i,j) \in \Ecal$ and $a_{ij} = 0$ otherwise. We denote, by $\boldsymbol{D} \in \R^{N\times N}$ its \emph{degree matrix}, i.e., $\boldsymbol{D} = \mathrm{diag}\left(\boldsymbol{d}\right)$ where $d_i = \sum_{j=1}^{N} a_{ij}$ for all $i\in \Vcal$. 
The graph $\Gcal$ is \emph{connected} if one can find a path in the graph connecting any two different agents.

Since $\boldsymbol{D}^{-1} \boldsymbol{A}$ is a stochastic matrix, from the Perron-Frobenius theorem, we have the following result.
\begin{lemma}[Perron-Frobenius]\label{lemma:perron_frobenius}
    Let $\Gcal$ be a connected graph. Then, the normalized adjacency matrix $\boldsymbol{D}^{-1} \boldsymbol{A}$ has a simple eigenvalue $1$, and all other eigenvalues have modulus strictly less than $1$. Moreover, the vector $\one$ is the right eigenvector associated with the eigenvalue $1$ of $\boldsymbol{D}^{-1} \boldsymbol{A}$.
\end{lemma}

Each agent $i \in \mathcal{V}$ is associated with a continuously evolving opinion $x_i(t) \in \mathcal{X}:= [-1, 1]$, representing their preference toward a certain behavior. Specifically, $x_i(t)$ is closer to $-1$ for pro-environmental attitudes and closer to $1$ for non-environmental (or unsustainable) tendencies. The overall state of opinions is collected in the vector $\boldsymbol{x}(t) = (x_i(t))_{i \in \mathcal{V}} \in \mathcal{X}^N$. 

The agents update their opinions by observing the behavior of their neighbors, which is captured by a continuously differentiable non-decreasing signal function $s: \mathcal{X} \to \mathcal{X}$. The collective perceived behaviors is given by the vector $\boldsymbol{s}(\boldsymbol{x})$ whose components are $s(x_i), i \in \mathcal{V}$. 

The environment is modeled as a state $\tilde{e}(t) \in \mathbb{R}_{\geq 0}$, capturing the collective impact of the network's influence. The behavior of each agent $i$ influences the environment through an increasing control function $u: \mathcal{X}^N \to \left[u_{\min}, u_{\max}\right]$, where $0 < u_{\min} < u_{\max} $. 
The environmental dynamics is:
\begin{equation*}
    \tau_{e} \dot{\tilde{e}} = -\gamma \tilde{e} + u(\boldsymbol{x}),
\end{equation*}
where $\gamma \in \left[0, 1\right]$ represents the environment’s natural recovery rate (in the absence of human action) and $\tau_{e} > 0$ is a time constant reflecting the speed of the environmental dynamics. 

Agents indirectly perceive the environmental state through a non-increasing function $r: \mathbb{R} \to \mathcal{X}$, which maps deviations from a threshold $\bar{e}$ into a response satisfying $r(\tilde{e} - \bar{e}) < 0$ if $\tilde{e} > \bar{e}$ and $r(\tilde{e} - \bar{e} ) > 0$ if $\tilde{e} < \bar{e}$. 
\begin{assumption}
	The threshold $\bar{e}$ is such that $\bar{e} \in \left[u_{\min}/\gamma, u_{\max}/\gamma\right]$.
\end{assumption}
This assumption ensures that the environment is at the same scale as the threshold. This allows agents to influence the environment in a meaningful way \cite{couthuresAnalysisOpinionDynamics2024}.

For analytical convenience, we recenter the environmental variable by writing $\tilde{e} = \bar{e} + e$, where $e \in \mathbb{R}$ represents the deviation from the threshold. The environmental dynamics is then reformulated as:
\begin{equation}\label{eq:environment_dynamics}
	\tau_{e} \dot{e} = -\gamma e  + u(\boldsymbol{x})- \gamma \bar{e}.
\end{equation}
Accordingly, $r(e) > 0$ when $e < 0$ (i.e., when the environment is better than the threshold) and $r(e) < 0$ if $e > 0$.

The opinion of each agent evolves according to the dynamics:
\begin{equation}\label{eq:opinion_dynamics}
    \tau_x \dot{x}_i = -x_i + \beta r(e) +  \frac{1-\beta}{d_i} \sum_{j=1}^{N}a_{ij} s(x_j),
\end{equation}
where $\tau_x > 0$ is the time constant for the opinion dynamics and $\beta \in \left[0, 1\right]$ quantifies the trade-off between environmental and social influences. In particular, the parameter $\beta$ plays a critical role: for $\beta < 0.5$, agents place more trust in the signals from their neighbors, whereas for $\beta > 0.5$ they give greater weight to the environmental signal. Similar OD models can be found in \cite{couthuresAnalysisOpinionDynamics2024,bizyaevaMultiTopicBeliefFormation2025,grayMultiagentDecisionMakingDynamics2018}. Such models can capture other behaviors than global synchronization like agreement, disagreement, clustering, and oscillations, 
making them encompass more realistic behaviors than the linear OD. 

The coupled system \eqref{eq:environment_dynamics} and \eqref{eq:opinion_dynamics} can be expressed in a matrix-vector form.
Defining the system state as $\boldsymbol{y}(t) = (\boldsymbol{x}(t), e(t))$, the dynamics can be written in a compact form:

\begin{equation}\label{eq:system_dynamics}
    \dot{\boldsymbol{y}}\! = \!\boldsymbol{F}(\boldsymbol{y})\! = \!\begin{bmatrix}
        \tau_x (-\boldsymbol{x} + \beta r(e)\mathbf{1} \! + \! (1-\beta)\boldsymbol{D}^{-1}\boldsymbol{A}\boldsymbol{s}(\boldsymbol{x})) \\
        \tau_e (-\gamma e  + u(\boldsymbol{x})- \gamma \bar{e})
    \end{bmatrix}\!.
\end{equation}

\section{Preliminary results}\label{sec:preliminary_results}
\subsection{Forward Invariance of Synchronization manifold}

We will now establish the forward invariance of the synchronization manifold
\begin{equation*}
    \mathcal{S} = \left\{ (\boldsymbol{x},e) \in \mathcal{X}^N \times \mathbb{R} \mid \boldsymbol{x} = p\one,\,  p \in \mathcal{X} \right\}.
\end{equation*}

\begin{proposition}
	Let $\Gcal$ be a connected graph. Then, the synchronization manifold $\mathcal{S}$ is forward invariant for \eqref{eq:system_dynamics}. 
\end{proposition}

\begin{proof}
	Let $(\boldsymbol{x}, e) \in \mathcal{S}$, i.e., $\boldsymbol{x} = p\one$ for some $p \in \mathcal{X}$. Then, for all $i,j \in \mathcal{V}$, one has 
	\begin{align*}
		\dot{x}_i - \dot{x}_j &= \tau_x \left( \boldsymbol{e}_i - \boldsymbol{e}_j \right)^\top \!\! \left( \beta r(e)\mathbf{1} + (1-\beta)\boldsymbol{D}^{-1}\!\boldsymbol{A}\boldsymbol{s}(\boldsymbol{x}) - \boldsymbol{x} \right)\\ 
		&= \tau_x (1-\beta) \left( \boldsymbol{e}_i - \boldsymbol{e}_j \right)^\top \!\! \left( \boldsymbol{D}^{-1}\!\boldsymbol{A}\boldsymbol{s}(p\one) - p\one \right)\\
		&= \tau_x (1-\beta) \left( s(p) - p \right) \left( \boldsymbol{e}_i - \boldsymbol{e}_j \right)^\top \! \one = 0.
	\end{align*}
	Since $\one$ is the eigenvector associated with the eigenvalue $1$ of $\boldsymbol{D}^{-1}\boldsymbol{A}$, from Lemma~\ref{lemma:perron_frobenius}. Thus, the synchronization manifold is forward invariant.
\end{proof}

The forward invariance of the synchronization manifold $\mathcal{S}$ ensures that once the system reaches the synchronization manifold, it remains there for all future times. Meaning that once all agents reach the same opinion, they will remain synchronized ever after. This property is crucial for the analysis of the system's behavior, as it allows us to focus on the dynamics over $\mathcal{S}$ and study the system's stability and bifurcations in a reduced space.

The attractiveness of the synchronization manifold $\mathcal{S}$ depends on the properties of the signal function $s$ and graph topology. For example, as shown in \cite[Proposition 3]{couthuresGlobalSynchronizationMultiagent2025}, for global underestimation function, i.e., $x (s(x) - x) \leq 0$ for all $x\in \left[-1,1\right]$, with any connected graph the synchronization manifold is globally asymptotically stable. For more results on the attractiveness and attraction basin of $\mathcal{S}$, we refer to \cite{couthuresGlobalSynchronizationMultiagent2025}, where synchronization for the dynamics \eqref{eq:opinion_dynamics} have been studied without environment coupling (i.e, for $\beta = 0$).

\subsection{Oddness of the dynamics}

\begin{assumption}\label{assumption:odd_dynamics}
	The functions $s$, $r$, and $u$ are smooth odd functions. Moreover, the control function $u$ satisfies $u(\boldsymbol{x} )  = -u(-\boldsymbol{x}) + 2\gamma \bar{e}$ for all $\boldsymbol{x} \in \mathcal{X}^N$. 
\end{assumption}

Under Assumption \ref{assumption:odd_dynamics}, the dynamics \eqref{eq:system_dynamics} is odd. Indeed, for all $\boldsymbol{y} = (\boldsymbol{x}, e) \in \mathcal{X}^N \times \mathbb{R}$, one has:
\begin{align*}
	\boldsymbol{F}(-\boldsymbol{y}) &= \begin{bmatrix}
		\tau_x \boldsymbol{f}_1(-\boldsymbol{x}, -e) \\
		\tau_e f_2(-\boldsymbol{x}, -e)
	\end{bmatrix} 	= \begin{bmatrix}
		-\tau_x \boldsymbol{f}_1(\boldsymbol{x}, e) \\
		-\tau_e f_2(\boldsymbol{x}, e)
	\end{bmatrix} = -\boldsymbol{F}(\boldsymbol{y}).
\end{align*}

This assumption is particularly useful. Indeed, in Section \ref{sec:analysis}, we will analyze the dynamics \eqref{eq:system_dynamics} of the system over $\mathcal{S}$ using the Lyapunov-Schmidt reduction method \cite{guckenheimerNonlinearOscillationsDynamical1983}. This method requires, in general, a lot of computational effort since it involves a Taylor expansion of the third order of the Jacobian matrix. Under Assumption~\ref{assumption:odd_dynamics}, the dynamics function exhibits an odd symmetry in the state variable, making the quadratic terms of this extension vanish and simplifying the analysis. Moreover, this assumption is also meaningful from a modeling perspective, as it aligns with the expectation that opposite behaviors occur around the neutral opinion and environment threshold.

\section{Analysis of synchronized agents dynamics}\label{sec:analysis}

In this section, we assume that the states of the agents are identical, meaning that $\boldsymbol{x}\in \mathcal{S}$, i.e., there exists a $p \in\mathcal{X}$ such that $\boldsymbol{x} = p \one$.
With a small abuse of notation, we will denote by $u(p)$ the function $u(p\one)$. The dynamics with opinion in $\mathcal{S}$ is then given by:
\begin{subequations}
\begin{eqnarray}\label{eq:FSO_dynamics}
		\tau_x \dot{p} &=& -p + \beta r(e) + (1-\beta) s(p) \\
		\tau_e \dot{e} &=& -\gamma e + u(p) - \gamma \bar{e}.\label{eq:E_dynamics}
\end{eqnarray}
\end{subequations}
We call \eqref{eq:FSO_dynamics} the Fully Synchronized Opinion (FSO) dynamics and denote $y = (p,e)$ its state.

Then, one can define the Fully Synchronized Opinion coupled with the Environment (FSOE) dynamics \eqref{eq:FSO_dynamics}-\eqref{eq:E_dynamics} through the function $F: \mathcal{X} \times \mathbb{R} \times [0, 1]^2 \to \mathcal{X} \times \mathbb{R}$ as:
\begin{equation}\label{eq:FSO_function}
	F(y,\beta, \gamma)=\begin{bmatrix}
		-p+\beta\,r(e)+(1-\beta)s(p) \\
		-\gamma\,e + u(p)-\gamma\bar{e}
	\end{bmatrix}.
\end{equation}

\subsection{Equilibria}

Let us define the instrumental function $g: \mathcal{X} \to \mathcal{X}$ as:
\begin{equation*}
	g(x) = \beta r\left(\frac{u(x)}{\gamma} - \bar{e}\right) + (1-\beta)s(x).
\end{equation*}

\begin{proposition}
	Let $\boldsymbol{y}^*$ be an equilibrium of the FSOE dynamics. Then, $\boldsymbol{y}^* = (p^*, u(p^*)/\gamma)$ where $p^*$ is a fixed point of $g$.
\end{proposition}

\begin{proof}
	Let $\boldsymbol{y}^* = (p^*, u(x^*)/\gamma - \bar{e})$ be an equilibrium of the FSO dynamics. Then, one has the following from the environmental dynamics:
	\begin{align*}
		\dot{e} = 0 \Leftrightarrow e^* = u(p^*)/\gamma - \bar{e}.
	\end{align*}
	Substituting this into the opinion dynamics yields:
	\begin{align*}
		&\dot{p} = 0 \Leftrightarrow p^* = \beta r(e^* ) + (1-\beta)s(p^*) \\
		&\Leftrightarrow p^* = \beta r\left(\frac{u(p^*)}{\gamma}-\bar{e}\right) + (1-\beta)s(p^*) \Leftrightarrow p^* = g(p^*).
	\end{align*}
\end{proof}

The equilibrium points of the FSOE dynamics are then given by the fixed points of the instrumental function $g$. The function $g$ may have multiple fixed points, leading to the existence of multiple equilibria of \eqref{eq:FSO_dynamics}-\eqref{eq:E_dynamics}. In the following, we focus our analysis at $y^*=(0,0)$, that is always an equilibrium for any $\beta$ and $\gamma$ under Assumption~\ref{assumption:odd_dynamics}.

\subsection{Singular points}
In this subsection, we analyze the conditions under which the Jacobian of the FSOE dynamics becomes singular. Identifying these singular points is essential as they mark the parameter values at which the linearization of the system loses full rank, thereby signaling potential bifurcations and qualitative changes in the system's behavior. 

Our approach is as follows. First, we derive the explicit form of the Jacobian matrix of \eqref{eq:FSO_function}. In particular, establishing conditions under which the leading (i.e. maximal real part) eigenvalues of the Jacobian lie on the imaginary axis
paves the way for rigorously proving the existence of pitchfork and Hopf bifurcations 
in later sections.

The Jacobian matrix of $F$ is given by:
\begin{equation*}
	D_{y}F(y,\beta, \gamma) \! = \! \begin{bmatrix}
		\tau_x^{-1}  \!\!\!& 0\\
		0 \!\!\!& \tau_e^{-1}
	\end{bmatrix}\!\! \begin{bmatrix}
		(1-\beta) s'(p)  - 1  \!\!& \beta r'(e)\\
		u'(p)\!\! & -\gamma
	\end{bmatrix}\!\!.
\end{equation*}

The following proposition provides the conditions under which the Jacobian matrix of the FSOE dynamics has singular points.  In the following, we will note $\tau = \tau_e / \tau_x$.

\begin{proposition}\label{prop:singular_points}
	The Jacobian matrix $D_{y}F(y,\beta,\gamma)$ has:
	\begin{enumerate}
		\item at least one eigenvalue equal to zero if and only if 
		\begin{equation}\label{eq:zero_eigenvalue_condition_gamma}
			\gamma = -\frac{u'(p)  r'(e) \beta}{(1-\beta) s'(p)  - 1},
		\end{equation}
		and $(1-\beta) s'(p)  > \max(1,-u'(p)  r'(e)\beta + 1)$.\label{itm:1}
		\item two real eigenvalues equal to zero if and only if the conditions of \ref{itm:1}) are satisfied and
		\begin{equation*}\label{eq:zero_eigenvalue_condition_beta}
			\beta = 1 - \frac{1}{s'(p)} +  \frac{- u'(p)  r'(e)  - \sqrt{\Delta_\beta}}{2\tau  s'(p)^2},
		\end{equation*}
		with $\Delta_\beta = u'(p) r'(e) [ u'(p) r'(e) - 4\tau s'(p) (s'(p) - 1)]$.
		\item two conjugated complex eigenvalues with to zero real part if and only if
		\begin{equation}\label{eq:hopf_condition}
			\gamma = \tau((1-\beta) s'(p)  - 1),
		\end{equation}
		with $1 < (1-\beta)s'(p) < 1 + 1/\tau$ and $\beta \in \left( \beta_{-}, \min(\beta_{+},1)\right)$, where 
		\begin{align*}
			\beta_{\pm} = 1 - \frac{1}{s'(p)} +  \frac{- u'(p)  r'(e)  \pm \sqrt{\Delta_{\beta}}}{2\tau  s'(p)^2}.
		\end{align*}\label{itm:3}
		Moreover, the eigenvalues are given by $\pm i\omega_0$ with $\omega_0 = \sqrt{\mathrm{det}(D_y F(y,\beta,\gamma))}$.
	\end{enumerate}
\end{proposition}

\begin{proof}
	We prove each of the three statements in turn.
	
	The Jacobian $D_{y}F(y,\beta,\gamma) := D_{y}F$ has a zero eigenvalue if and only if its determinant vanishes. A direct computation yields
	\begin{equation*}
		\det\!\left(D_{y}F\right) = - \frac{\gamma\left(\left(1-\beta\right) s'(p)  - 1\right) + u'(p)  r'(e)\beta}{\tau_x \tau_e}.
	\end{equation*}
	Thus, setting $\det\!\left(D_yF\right)=0$ one obtains
	\begin{equation*}
		\gamma = -\frac{\beta\,u'(p)\,r'(e)}{(1-\beta)s'(p)-1}.
	\end{equation*}
	This relation is meaningful provided that $(1-\beta) s'(p)  > \max(1,-u'(p)  r'(e)\beta+1)$ since $u'(p) r'(e) \leq 0$ and $\gamma \in \left[0, 1\right]$. 
	
	For the Jacobian to have a double zero eigenvalue, both the determinant and the trace must vanish. The trace is
	\begin{equation*}
		\mathrm{Tr}\!\left(D_{y}F\right) = \frac{(1-\beta) s'(p)  - 1}{\tau_x} - \frac{\gamma}{\tau_e}.
	\end{equation*}
	Then, the Jacobian has two eigenvalues equal to zero if and only if its trace and the determinant are equal to zero. This yields the following system of equations:
	\begin{equation*}
		\gamma = -\frac{u'(p)  r'(e)\beta}{(1-\beta) s'(p)  - 1}\text{ and }\gamma = \tau ((1-\beta) s'(p)  - 1).
	\end{equation*}
	Then, one can retrieve the parameter $\beta$ as the solution of the following equation:
	\begin{align}
		& \tau ((1-\beta) s'(p)  - 1)^2 + u'(p)  r'(e)\beta = 0\nonumber\\
		\Leftrightarrow {}&{} \tau  \beta^2 s'(p)^2 + \Big( u'(p)  r'(e) - 2\tau   s'(p) \left(s'(p) - 1\right)\Big)\beta\nonumber \\
		& \qquad + \left(s'(p)-1\right)^2 = 0.\label{eq:polynomial_beta}
	\end{align}
	Denoting $\allowbreak\Delta_{\beta} = u'(p)  r'(e) [ u'(p)  r'(e) - 4\tau   s'(p)\allowbreak (s'(p) - 1)]$, one has that the equation has a real solution if and only if $u'(p)  r'(e) \leq  4\tau   s'(p)\allowbreak (s'(p) - 1)$. Under the assumption $u'(p)r'(e)\le 0$ and $s'(p)>1$ (implied by the condition in (1)), a real solution exists. The two solutions are given by
	\begin{equation*}
		\beta_{\pm} = 1 - \frac{1}{s'(p)} +  \frac{- u'(p)  r'(e)  \pm \sqrt{\Delta_{\beta}}}{2\tau  s'(p)^2}.
	\end{equation*}

	A further inspection shows that only \(\beta_{-}\) satisfies the additional requirement \((1-\beta)s'(p) > 1\).

	Finally, the Jacobian has two complex conjugate eigenvalues with zero real part if and only if (i) the trace vanishes and (ii) the discriminant of the characteristic polynomial is negative. The characteristic polynomial is
	\begin{align*}
		P(X) &= X^2 - \mathrm{Tr}\!\left(D_{y}F\right) X  + \det\!\left(D_{y}F\right).
	\end{align*}
	with discriminant
    \begin{align}\label{eq:discriminant_characteristic_polynomial}
		\Delta_{P} &= \mathrm{Tr}\!\left(D_{y}F\right)^2 -4\det\!\left(D_{y}F\right).
	\end{align}
	Setting $\mathrm{Tr}(D_yF)=0$ again yields
	\begin{equation*}
		\gamma = \tau((1-\beta) s'(p)  - 1),
	\end{equation*}
	Moreover, to ensure that the eigenvalues are non-real, we require $\Delta_P<0$, which is equivalent to $\det(D_yF)>0$. This inequality leads to
	\begin{equation*}
		\gamma\left(\left(1-\beta\right) s'(p)  - 1\right) + u'(p)  r'(e)\beta <0.
	\end{equation*}
    
	In light of our previous derivations, this condition is satisfied for $\beta \in (\beta_{-},\min\{\beta_{+},1\})$, with the additional constraint that $ 1 < (1-\beta)s'(p) < 1 + 1 / \tau$. Moreover, the eigenvalues are given by $\pm i\omega_0$, where $\omega_0 = \sqrt{\det\!\left(D_yF(y,\beta,\gamma)\right)}$.
\end{proof}

This proposition provides a complete characterization of the singular point of the Jacobian matrix of the FSOE dynamics \eqref{eq:FSO_dynamics}. 

\begin{remark}
    It is noteworthy that the condition $s'(p) > 1$ is necessary for the Jacobian of the FSOE dynamics at $y=(p,e)$ to become singular. In other words, at the opinion state $p$, the signal function $s$ must act as an amplifier of the agents' opinions for it to be singular.
\end{remark}

\subsection{Bifurcation analysis}

The following result provides conditions for a pitchfork bifurcation at the equilibrium $(p,e)=(0,0)$ of \eqref{eq:FSO_dynamics}-\eqref{eq:E_dynamics}. 

\begin{theorem}\label{prop:bifurcation}
	Suppose Assumption~\ref{assumption:odd_dynamics} holds and that $\beta^*$ and $\gamma^* = \gamma(\beta^*)$ satisfy \eqref{eq:zero_eigenvalue_condition_gamma} at $(p,e)=(0,0)$. 
    Then, the Jacobian of the FSOE dynamics~\eqref{eq:FSO_dynamics}-\eqref{eq:E_dynamics} at the equilibrium $y^* = (0,0) $ has a zero eigenvalue associated with the critical eingenvector
    \begin{equation*}
        	v = \begin{bmatrix} 1 & u'(0) / \gamma^* \end{bmatrix}^\top.
    \end{equation*}
    Moreover, if 
	\begin{equation*}
		c=(1-\beta^*)s'''(0)+\beta^*r'''(0)\frac{u'(0)^3}{{\gamma^*}^3} + \frac{\beta^*r'(0)}{\gamma^*} u'''(0)  
	\end{equation*}
    is nonzero, then a pitchfork bifurcation occurs at $y^*$ as $\beta$ passes through $\beta^*$. In particular, the bifurcating branches emerge along the subspace generated by $v$.

\end{theorem}

\begin{proof}
From proposition \ref{prop:singular_points},  when $\beta^*$ verifies \eqref{eq:zero_eigenvalue_condition_gamma}, we know that the Jacobian of the FSOE dynamics \eqref{eq:FSO_dynamics}-\eqref{eq:E_dynamics} at the equilibrium $y^* = (0,0) $ is singular. Inspired by \cite{golubitskySingularitiesGroupsBifurcation1985}, we perform the Lyapunov-Schmidt reduction of the equilibrium equation $F(y, \beta^*,\gamma^*)=0$. At $(y^*, \beta^*,\gamma^*)$, the linearization $D_{y}F$ is singular and has a one-dimensional kernel. Let $v$ and $w$ be the corresponding right and left eigenvectors associated with the zero eigenvalues. One has 
	\begin{equation*}
		v = \begin{bmatrix}
			1 & u'(0)/\gamma^*
		\end{bmatrix}^{\top} \quad \text{and} \quad w = \begin{bmatrix}
			1 & \beta^* r'(0)/\gamma^*
		\end{bmatrix}.
	\end{equation*}
	Then, $P = w w^{\top}/ \norm{w}$ is the projection onto the kernel of $D_{y}F(y^*, \beta^*,\gamma^*)$ and $Q = I - P$ is the projection onto its range. One can write a general nearby solution as $y = z  v + \xi w$,
	where $z  \in \R$ the coordinate along the kernel and $\xi \in \R$ the coordinate along the range. 
	Let $\mu = \beta - \beta^*$ denote the unfolding parameter. The equilibrium equation $F(y, \beta^*,\gamma^*)=0$ is then equivalent to the system

	$P F(y, \beta^* + \mu, \gamma^*) = 0$ and $Q F(z  v + \xi w, \beta^* + \mu, \gamma^*) = 0$.

	By construction, one has $Q F(y^*,\beta^*,\gamma^*) = 0$. Since $Q D_{y}F(y^*,\beta^*,\gamma^*)$ is invertible, the implicit function theorem ensures that there exists a unique function $h(z ,\mu)$ such that $Q F(z  v + h(z ,\mu) w, \beta^* + \mu, \gamma^*) = 0$ and $h(0,0) = 0$. Then, a nearby solution to the equilibrium equation is given by $y=z  v+h(z ,\mu)$.
	
	Substituting this ansatz into $F(y, \beta^* + \mu, \gamma^*)=0$ and projecting onto the kernel (by multiplying on the left by $w^{\top}$) yields the reduced bifurcation equation $\Phi(z ,\mu) = 0$.

	Since the functions $s$, $r$, and $u$ are odd, their Taylor expansions about the origin contain only odd-order terms. This symmetry ensures that no quadratic term appears in the expansion of $\Phi$ in $z $. Thus, expanding in powers of $z $ and $\mu$ we obtain
	\begin{equation*}
		\Phi(z ,\mu)=a\mu z +\frac{1}{6}c z ^3+\mathcal{O}(z^5,\mu z^3,\mu^2 z )=0,
	\end{equation*}

	with some constant 
	\begin{equation*}
		a = w^\top \partial_{\beta} F(y^*,\beta^*,\gamma^*) = w^\top \begin{bmatrix}
			\frac{r(0) - s(0)}{\tau_x} & 0 
		\end{bmatrix}^\top = 0,
	\end{equation*}
	since $r(0) = s(0) = 0$.
	The third-order terms are given by:
	\begin{equation*}
		D^3_{y}F(y^*\!\!, \beta^*\!\!, \gamma^*)(h,k,l) \! = \! \begin{bmatrix}
			\tau_x^{-1}  \!\!\!\!& 0\\
			0 \!\!\!\!& \tau_e^{-1}
		\end{bmatrix} \!\! \begin{bmatrix}
			\begin{split}
				&(1-\beta^*) s'''(0) h_1 k_1 l_1  \\
				&\,\, \,  \quad+ \beta^* r'''(0) h_2 k_2 l_2 
			\end{split}	\\
			u'''(0) h_1 k_1 l_1 
		\end{bmatrix} \!\!
	\end{equation*}
	for $y = (p, e)$, $h = (h_1, h_2)$, $k = (k_1, k_2)$ and $l = (l_1, l_2)$. 

	Then, the cubic coefficient $c$ is given by
	\begin{align*}
		&c=w^\top \left[D^3_{y}F(y^*, \beta^*,\gamma^*)(v,v,v)\right]\\
	&=\frac{(1-\beta^*)s'''(0)}{\tau_x}+\frac{\beta^*r'''(0)u'(0)^3}{\tau_x {\gamma^*}^3} + \frac{\beta^*r'(0)u'''(0)}{\tau_e \gamma^*}.
	\end{align*}

Consequently, if $c\neq 0$ a pitchfork bifurcation occurs at $y^*$ as $\beta=\beta^*$. Precisely, two nontrivial symmetric solutions bifurcate from $z =0$, along the $v$ direction, when the sign of $\mu$ changes.
\end{proof}

From an environmental view-point,  this pitchfork bifurcation captures the idea that small changes in the trust parameter $\beta$ can trigger sudden transitions between collective opinion states. Close to a pitchfork bifurcation, dynamics are bistable: the state can converge to either of two stable equilibria, depending on perturbations or initial conditions. In the environmental setting, either the collective opinion strongly favors pro-environmental actions, resulting in a well-preserved environment, or the prevailing sentiment opposes environmental efforts, leading to environmental degradation.

\begin{remark}
Although the pitchfork of Theorem~\ref{prop:bifurcation} happens at an unstable trivial equilibrium, it is this bifurcation that gives rise to bistability in the FSOE dynamics. 
\end{remark} 

In addition to the pitchfork bifurcation, the trivial equilibrium also exhibits a Hopf bifurcation for a different range of parameter $\beta$.

\begin{theorem}\label{prop:hopf}
Suppose Assumption~\ref{assumption:odd_dynamics} holds and let $\beta^*$ and $\gamma^*= \gamma(\beta^*)$ satisfy \eqref{eq:hopf_condition} with equilibrium $y^* =(0,0)$. Then, a Hopf bifurcation yielding a family of periodic orbits occurs at $y^*$ as $\beta=\beta^*$. Moreover, the bifurcating limit cycle is unique and stable for $\beta < \beta^*$ if the coefficient $h_{21}$ defined in \eqref{h21} satisfies $\mathrm{Re}(h_{21})\neq 0$.	

\end{theorem}

\begin{proof}
	By \cite[Theorem 3.4.2]{guckenheimerNonlinearOscillationsDynamical1983}, the existence of periodic orbits follows if the system \eqref{eq:FSO_dynamics}-\eqref{eq:E_dynamics} satisfies two conditions at the equilibrium $y^* = (0,0)$: the eigenvalues of the Jacobian matrix $D_{y}F(y^*,\beta^*,\gamma^*)$ are purely imaginary, and the derivative of the real part of the eigenvalues with respect to $\beta$ is nonzero. 

	First, by Proposition \ref{prop:singular_points}, the eigenvalues of the Jacobian matrix at $y^* = (0,0)$ are purely imaginary, satisfying the first condition. Moreover, in a neighborhood of $(y^*,\beta^*,\gamma^*)$, the characteristic equation of the Jacobian $D_y F$ has discriminant \eqref{eq:discriminant_characteristic_polynomial} with $\Delta_{P} < 0$, giving eigenvalues
	\begin{equation*}
		\lambda = \frac{\mathrm{Tr}\!\left(D_{y}F(y,\beta, \gamma^*)\right) + i\sqrt{- \Delta_P(y,\beta, \gamma^*)}}{2},
	\end{equation*}
	with $\Delta_P$ given by \eqref{eq:discriminant_characteristic_polynomial}. 

	The second condition is also satisfied since the real part of these eigenvalues is $\mathrm{Re}(\lambda) = \mathrm{Tr}(D_{y}F(y,\beta, \gamma^*)) / 2$. Differentiating with respect to $\beta$, we obtain $\partial_\beta\mathrm{Re}(\lambda) = -s'(0) / \tau_x$, which is nonzero since $s'(0) > 1$. This proves the existence of the periodic orbits.

    The stability and uniqueness of the limit cycle results from \cite[Theorem 3.3]{kuznetsovElementsAppliedBifurcation}. First, let us provide the normal form of the bifurcation. At $(y^*, \beta^*,\gamma^*)$, the linearization $D_{y}F$ is singular and has a two-dimensional center subspace associated with the purely imaginary eigenvalues $\pm i \omega_0$. Let $q$ and $\bar{q}$ be the right eigenvectors and $p$ and $\bar{p}$ be the left eigenvectors associated with the eigenvalues $i\omega_0$ and $-i\omega_0$, respectively. They are given by
	\begin{equation*}
		q = \begin{bmatrix}
			\tau_x^{-1} \beta^* r'(0) \\
			 i \omega_0 - a
		\end{bmatrix} \quad \text{and} \quad p = 
        \begin{bmatrix}
			\tau_e^{-1} u'(0) & i \omega_0 - a
		\end{bmatrix},
	\end{equation*}
	where $a = \tau_x^{-1} ((1-\beta^*)s'(0) - 1) = \tau_x^{-1} \gamma^*$. 

	By \cite[Lemma 3.3]{kuznetsovElementsAppliedBifurcation}, for $\beta$ sufficiently close to $\beta^*$ and setting $\mu = \beta - \beta^*$, the FSOE dynamics \eqref{eq:FSO_dynamics}-\eqref{eq:E_dynamics} can be transformed via a complex variable $z$ into
	\begin{equation}\label{eq:expression_z_original}
		\dot{z} = \lambda z + h(z, \bar{z}, \mu),
	\end{equation}
	with $h(z, \bar{z}, \mu) = \mathcal{O}(\abs{z}^2)$ a smooth function of $z$ and $\bar{z}$.

    Expanding $h$ in powers of $z$ and $\bar{z}$,
	\[ h(z,\bar{z},\mu) = \sum_{\mathclap{k+l \ge 2\,}} \frac{1}{k!l!} \frac{\partial^{k+l}}{\partial z^k \,\partial \bar{z}^l} \left\langle p, F\bigl(zq + \bar{z}\bar{q}, \mu\bigr)	\right\rangle_{{|{z=0}}}  z^k \bar{z}^l. \]

    By \cite[Lemma 3.6]{kuznetsovElementsAppliedBifurcation}, equation \eqref{eq:expression_z_original} can be rewritten as $\dot{z} = \lambda z + c_1 z^2 \bar{z} + \mathcal{O}(\abs{z}^4)$, where $c_1$ is the first Lyapunov exponent given by
	\begin{equation*}
		c_1 = \frac{h_{20}h_{11} (2\lambda + \bar{\lambda})}{2\abs{\lambda}} + \frac{\abs{h_{11}}^2}{2} + \frac{\abs{h_{02}}^2}{2(2\lambda - \bar{\lambda})}  + \frac{h_{21}}{2}.
	\end{equation*}
	
Since $s$, $r$ and $u$ are odd functions, one has that $h_{kl} = 0$ for $k + l$ even, leading to $c_1 = h_{21}/2$.
The coefficient $h_{21}$ is given by
	\begin{align}\label{h21}
		h_{21} &= \bar{p}^\top D^3_{y}F(y^*, \beta^*,\gamma^*)(q,q,\bar{q}) \notag\\
		&= u'(0)\biggl[ \frac{(1-\beta^*) s'''(0) \left(\beta^* r'(0)\right)^3}{\tau_x^3} + \\
			\beta^* r'''(0)&(i \omega_0 + a)(i \omega_0 - a)^2\biggr]\! - \frac{u'''(0)\! \left(i \omega_0 + a\right)\!\left(\beta^* r'(0)\right)^3 \!\!}{\tau_x^2} ,\notag
	\end{align}
with $a=\tau_x^{-1}\gamma^*$ and $\omega_0=\sqrt{\det\!\left(D_y F(y^*,\beta^*,\gamma^*)\right)}$.
If $\mathrm{Re}(h_{21}) \neq 0$ one has $\mathrm{Re}(c_{1})\neq 0$ and from \cite[Theorem 3.3]{kuznetsovElementsAppliedBifurcation} there is a unique stable limit cycle that bifurcates from the equilibrium $y^*$ via a Hopf bifurcation for $\beta < \beta^*$.
\end{proof}

In summary, our bifurcation analysis reveals two distinct types of qualitative transitions in the coupled opinion-environment dynamics. The pitchfork bifurcation indicates a sudden, symmetry-breaking shift in collective behavior, whereas the Hopf bifurcation signals the emergence of oscillatory dynamics, which may model the recurrent cycles of environmental collapse and recovery observed in prey-predatory systems.

\section{NUMERICAL SIMULATIONS} \label{sec:numerical_simulations}

We illustrate the results of Theorem~\ref{prop:bifurcation} and \ref{prop:hopf} with numerical simulations. We consider the following functions $s(x) = \tanh(3x)$, $r(x) = \tanh(-3x)$ and $u(x) = x -\gamma \bar{e}$, where $\bar{e} = 0.5$ is the environmental threshold.
We set $\tau_x = \tau_e = 1$ and $\gamma = 0.2$. We compute the bifurcation diagram for the FSO dynamics \eqref{eq:FSO_dynamics} with respect to the parameter $\beta$ in Figure~\ref{fig:bifurcation_diagram}. We observe a pitchfork bifurcation around $\beta = 0.24$ and a Hopf bifurcation around $\beta = 0.60$ in Figures \ref{fig:first_bifurcation} and \ref{fig:hopf_bifurcation}, respectively.

\begin{figure}[t]
    \centering
    \includegraphics[width=\linewidth]{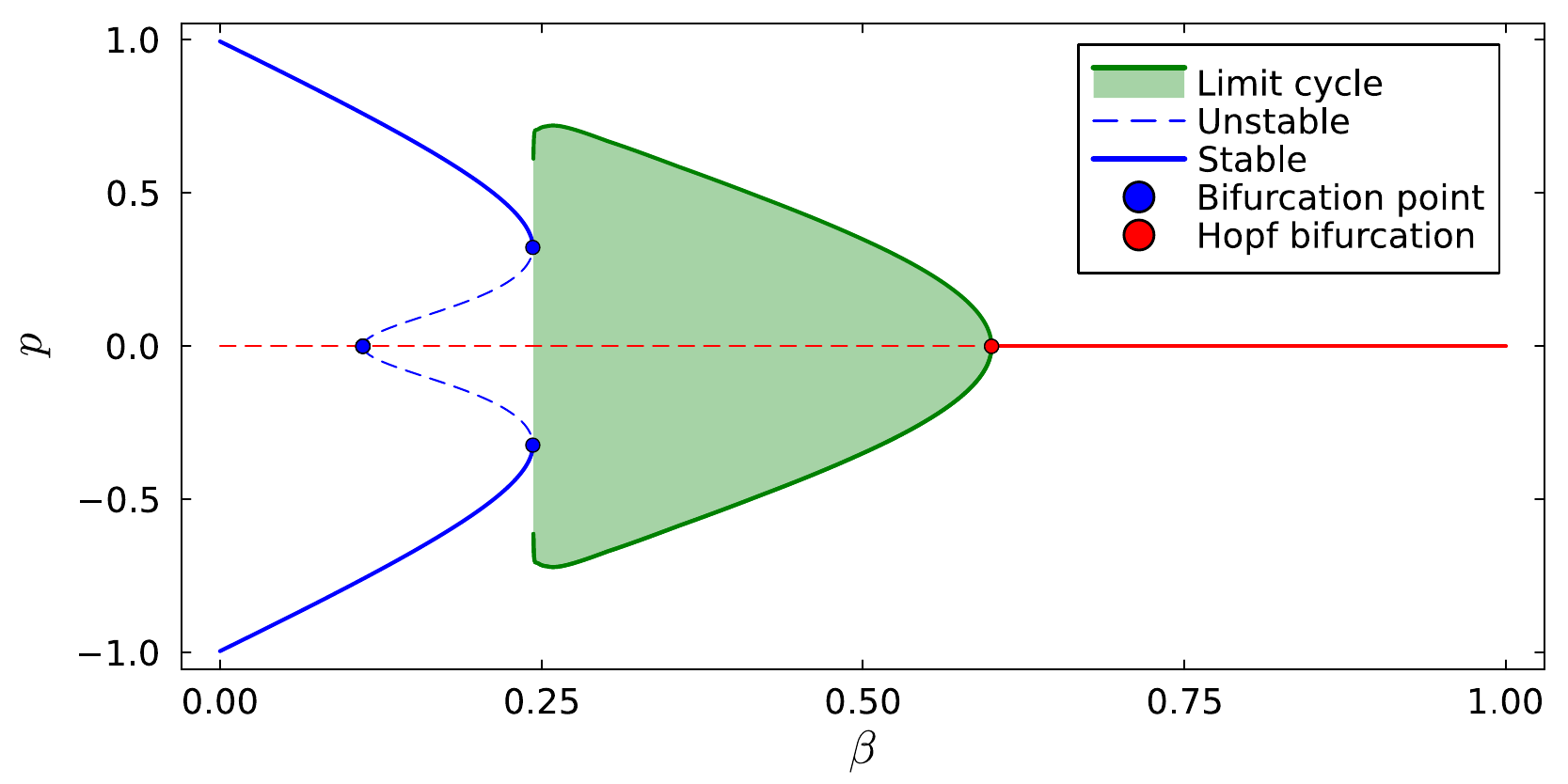}
    \caption{Bifurcation diagram for the FSO dynamics \eqref{eq:FSO_dynamics} with $\gamma = 0.2$.  
    Solid lines denote stable equilibria, dashed lines denote unstable equilibria, dots mark bifurcation points, and the green region shows the amplitude of the stable limit cycle emerging from the Hopf bifurcation.}
    \label{fig:bifurcation_diagram}
    \vspace{-0.3cm}
\end{figure}

\begin{figure}[t]
    \centering
    \begin{subfigure}[t]{0.49\linewidth}
        \includegraphics[width=\linewidth]{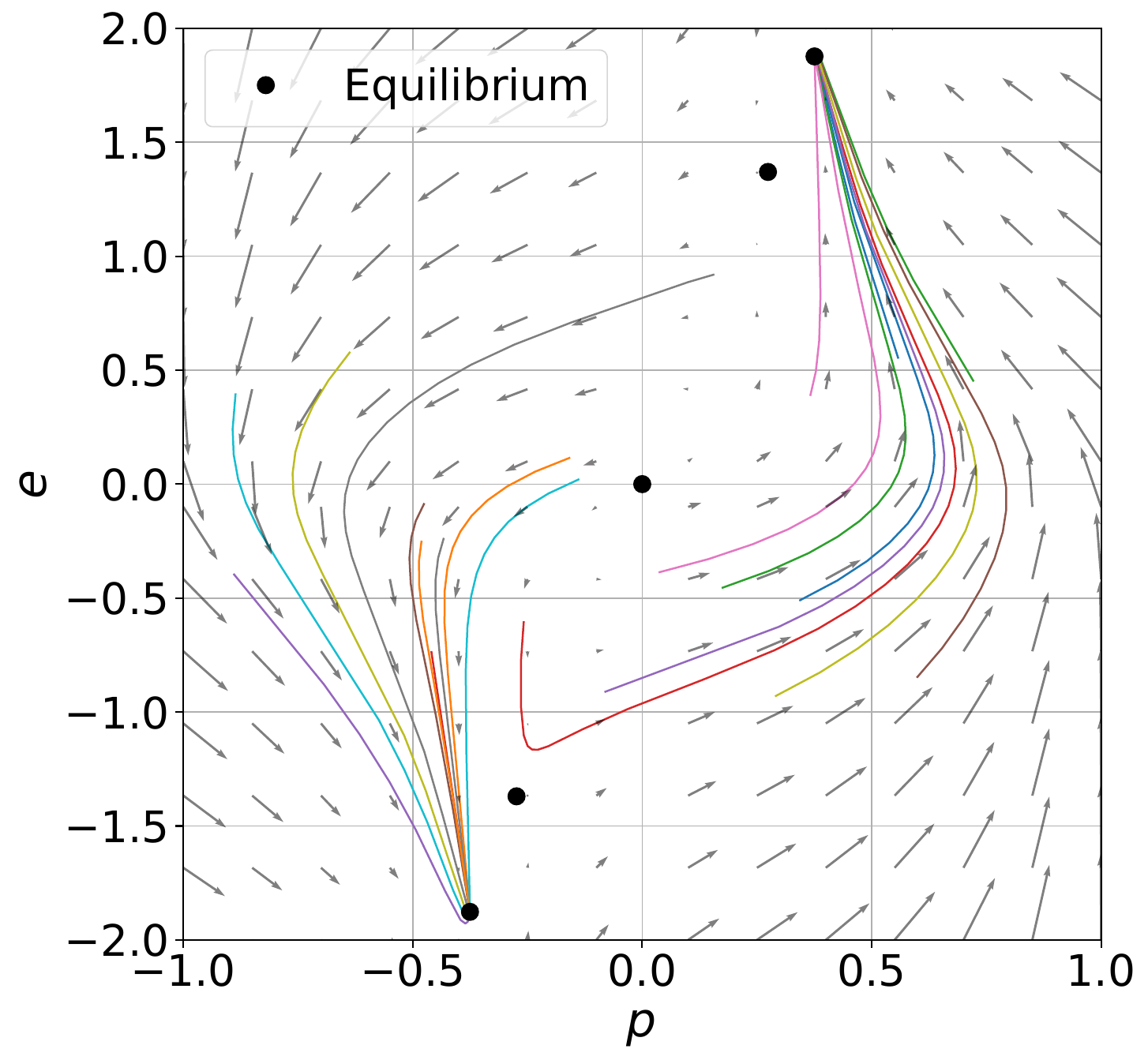}
        \caption{$\beta = 0.24$}
    \end{subfigure}
    \hfil
    \begin{subfigure}[t]{0.49\linewidth}
        \includegraphics[width=\linewidth]{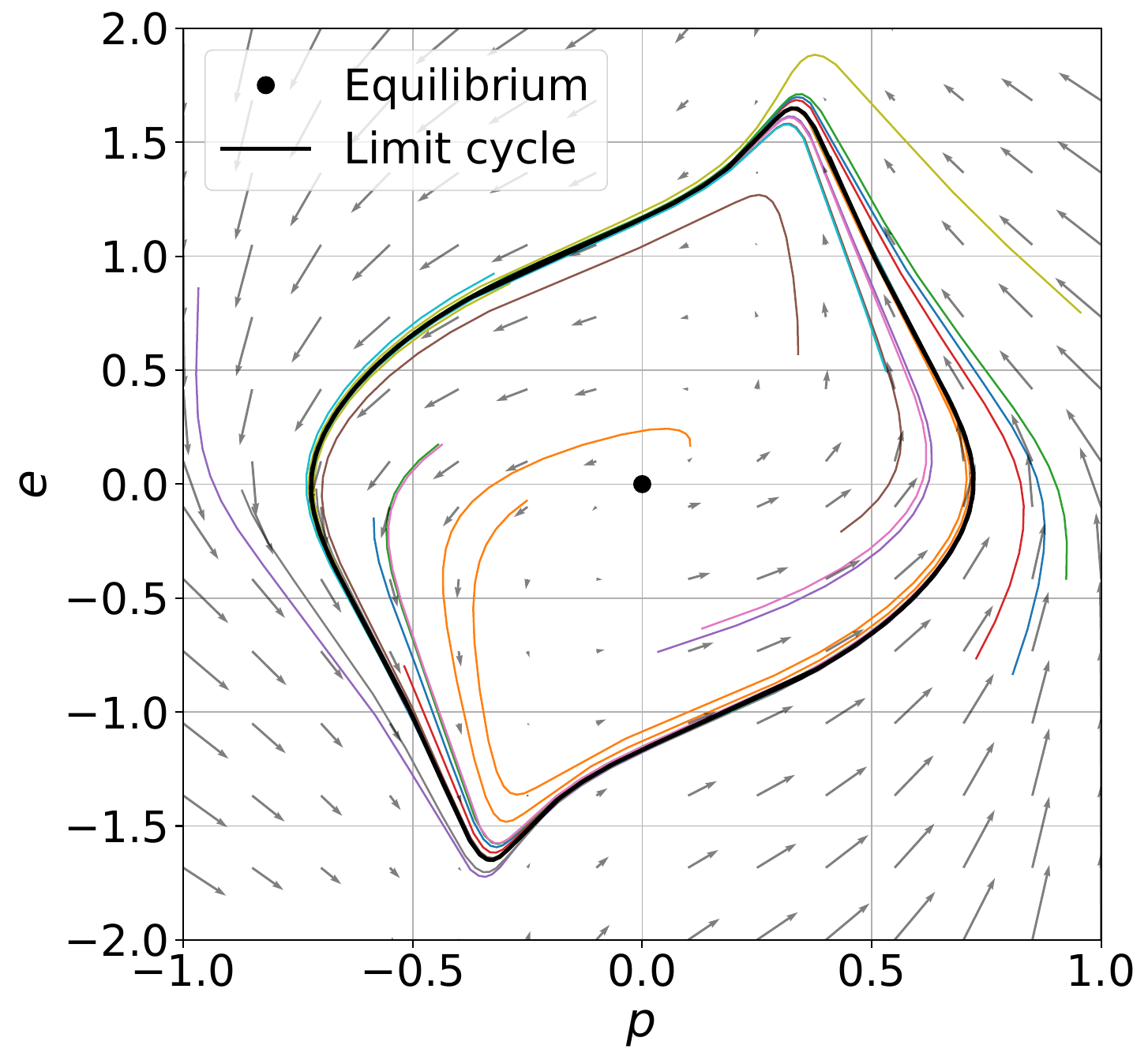}
        \caption{$\beta = 0.25$}
    \end{subfigure}
    \caption{Phase portraits for two value of $\beta$ and $\gamma = 0.2$. The system passes from 5 equilibria for $\beta = 0.24$ to 1 equilibrium with a stable limit cycle for $\beta =0.25$.}
    \label{fig:first_bifurcation}

    \vspace{-0.3cm}
\end{figure}

\begin{figure}[t]
    \centering
    \begin{subfigure}[t]{0.49\linewidth}
        \includegraphics[width=\linewidth]{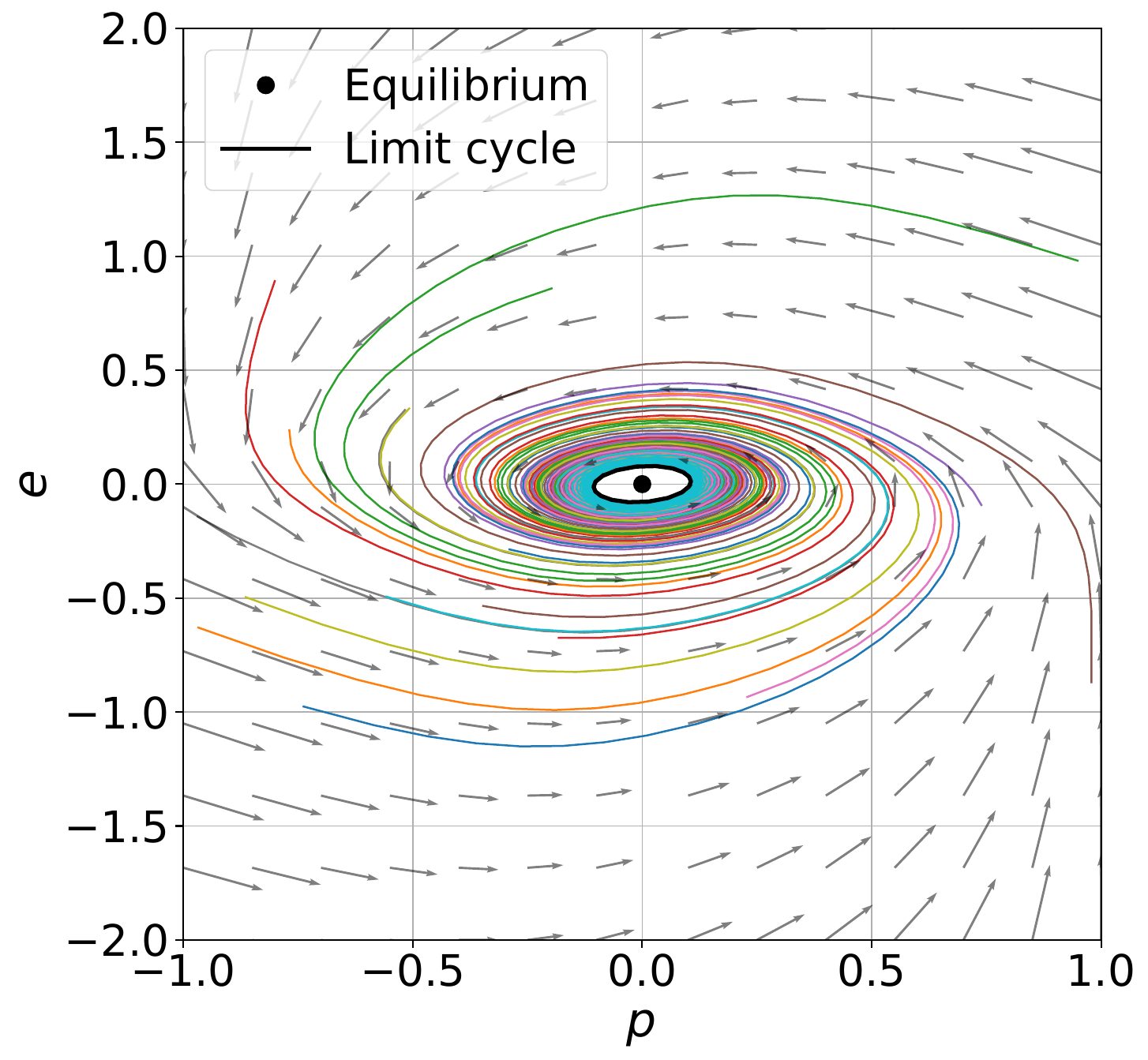}
        \caption{$\beta = 0.59$}
    \end{subfigure}
    \hfil
    \begin{subfigure}[t]{0.49\linewidth}
        \includegraphics[width=\linewidth]{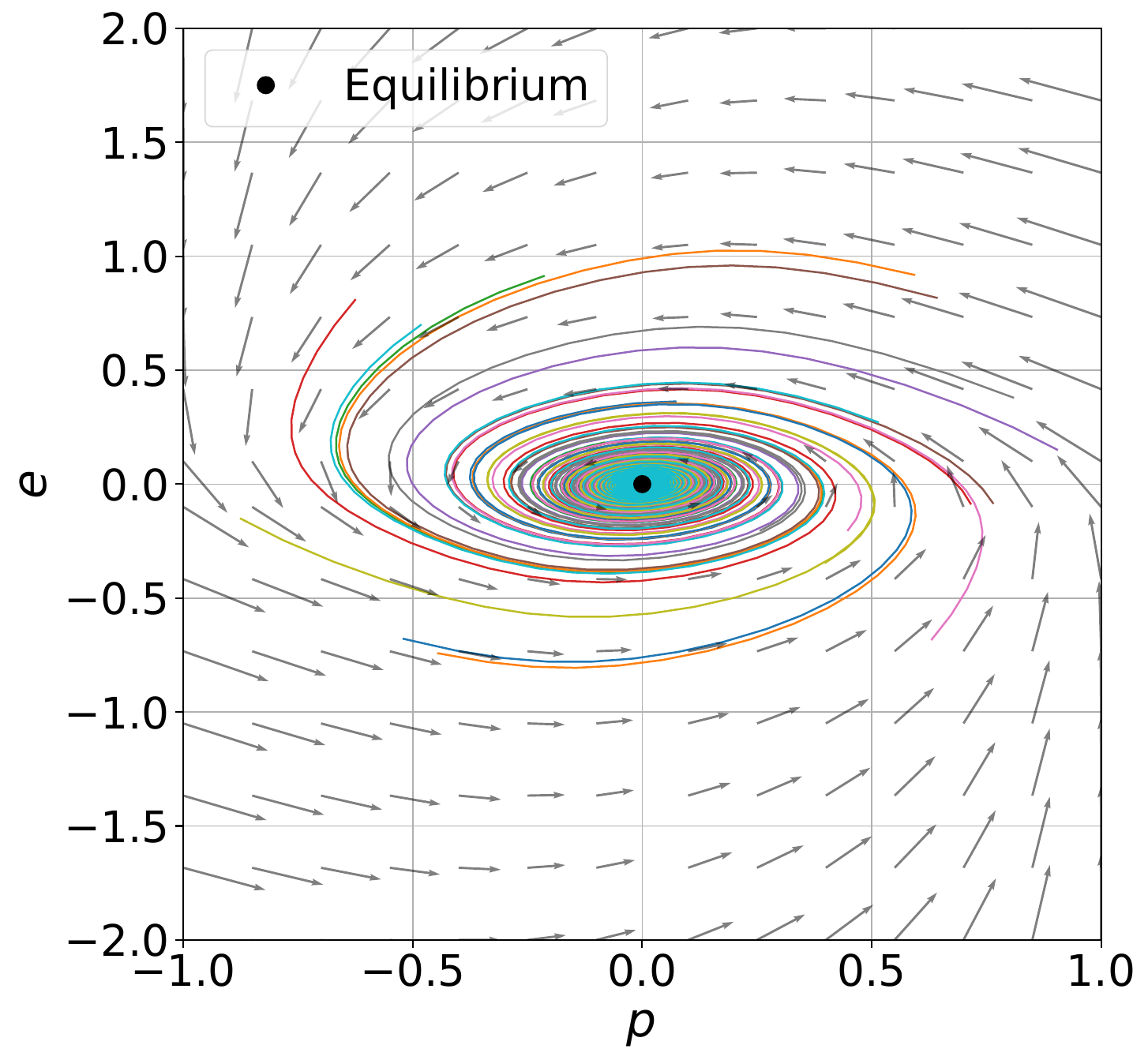}
        \caption{$\beta = 0.61$}
    \end{subfigure}
    \caption{Phase portraits for two value of $\beta$ and $\gamma = 0.2$. The stable limit cycle for $\beta=0.59$ collapses into an equilibrium at $\beta = 0.61$ illustrating the Hopf bifurcation of Figure~\ref{fig:bifurcation_diagram}.}
    \label{fig:hopf_bifurcation}
    \vspace{-0.3cm}
\end{figure}

\begin{remark}
    As illustrated in Figure~\ref{fig:bifurcation_diagram}, the system also exhibits saddle-node bifurcation involving the equilibria that emerge from the pitchfork bifurcation of Theorem~\ref{prop:bifurcation}. Due to lack of space,  we omit a detailed analysis of these saddle-node bifurcations; such an analysis could be carried out using techniques analogous to those in Theorem~\ref{prop:bifurcation}.
\end{remark}

\section{CONCLUSIONS} \label{sec:conclusions}

In this paper, we introduce and analyze a continuous-time opinion-environment model as an extension of \cite{couthuresAnalysisOpinionDynamics2024}, capturing the interplay between social interactions and environmental feedback. We establish positive invariance, explore singularities of the FSOE dynamics, and demonstrate pitchfork and Hopf bifurcations using a rigorous mathematical framework.

\addtolength{\textheight}{-12cm}   





\bibliographystyle{IEEEtran}
\bibliography{Biffurcation_OD}

\end{document}